\documentclass{article}
\usepackage[american]{babel}
\usepackage[boxruled,vlined]{algorithm2e}
\usepackage{a4wide}
\usepackage{amsmath}
\usepackage{amssymb}
\usepackage{amsthm}
\usepackage{pxfonts}

\usepackage{sfmath}
\usepackage{graphicx}
\usepackage[T1]{fontenc}

\newcommand{\busy}[4]{
\begin{picture}(#4,3)
\put(0,0){\makebox(1,1)[c]{#1}}
\put(#4,0){\makebox(-1,1)[c]{#2}}
\put(0,1){\framebox(#4,1){}}
\put(0,2){\makebox(#4,1){#3}}
\end{picture}
}

\newcommand{\request}[2]{
\begin{picture}(4,4)
\put(0,3){\vector(0,-1){1}}
\put(0,3){\makebox(4,1)[l]{#1[#2]}}
\end{picture}
}

\newcommand{\requestsymb}{
\begin{picture}(1,1)
\put(0.5,1){\vector(0,-1){1}}
\end{picture}
}

\newcommand{\mkf}{$(m,k)$--firm}

\newcommand{\drop}[1]{}
\newtheorem{Theorem}{Theorem}
\newtheorem{Observation}{Observation}
\newtheorem{Definition}{Definition}

\newtheorem{Corollary}[Theorem]{Corollary}

\newcommand{\equals}{\stackrel{\mathrm{def}}{=}}
\newcommand{\lcm}{\mathrm{lcm}}

\newcommand{\bigOh}{\mathcal{O}}
\newtheorem{Example}{Example}

\begin{document}

\title{\mkf{} constraints and DBP scheduling: impact of the initial $k$-sequence and exact schedulability test}

\date{\today}

\author{
Jo\a"el Goossens\\Universit\a'e Libre de Bruxelles (U.L.B.)\\Brussels, Belgium\\joel.goossens@ulb.ac.be
}

\maketitle
\thispagestyle{empty}

\begin{abstract}
In this paper we study the scheduling of $(m,k)$--firm synchronous periodic task systems using the Distance Based Priority (DBP) scheduler. We first show three phenomena: (i) choosing, for each task, the initial $k$-sequence $1^k$ is not optimal, (ii) we can even start the scheduling from a (fictive) error state (in regard to the initial $k$-sequence) and (iii) the period of feasible DBP-schedules is not necessarily the task hyper-period. We then show that any feasible DBP-schedule is periodic and we upper-bound the length of that period. Lastly, based on our periodicity result we provide an exact schedulability test.
\end{abstract}


\section{Introduction}\label{intro} 
In this paper we consider the scheduling of $(m,k)$--firm real-time periodic task sets. The $(m,k)$--firm model was initially introduced by {Hamdaoui} et al.~\cite{Hamdaoui1995}; this model is intermediary between \emph{hard} real-time constraints (where deadline misses are fatal) and \emph{soft} real-time constraints (where deadline misses are tolerated but minimized). Indeed, the $(m,k)$--firm constraint imposes that ``few'' request deadlines can be missed. More formally, for each task, the $(m,k)$--firm constraint is characterized by two parameters: $m$ and $k$, and the constraint requires that at least $m$  requests meet their deadline for any $k$ consecutive requests of the task. This model is very often used to handle the scheduling of messages on real-time networks, the tasks represent the handling of network frames and are consequently and inherently \emph{non-preemptive} in the model. Hence, in the following, the term task denotes either a non-preemptive task (scheduled on a \emph{uni}processor) or a message (scheduled on a network channel). The seminal paper of {Hamdaoui} et al.~\cite{Hamdaoui1995} defined also a scheduling algorithm for that kind of constraints: the Distance Based Priority (DBP in short). Under DBP, the priority of each task is dynamic and is based on the number of task request[s] which can miss their deadline without violating the $(m,k)$--firm constraint. 

\paragraph{Related research.} 
Since the seminal paper of \mbox{Hamdaoui} et al.~\cite{Hamdaoui1995}, there is a huge literature about scheduling firm real-time systems for various kinds of schedulers. It is out of the scope of this paper to summarize that literature. Our study concerns a specific task model ($(m,k)$--firm periodic tasks) for a specific scheduler (DBP). Concerning the more specific case of $(m,k)$--firm constraints and the scheduler DBP, the literature proposes sufficient schedulability conditions~\cite{Li2004}. \mbox{Hamdaoui} in~\cite{Hamdaoui1997} evaluated the probability failure; \mbox{Lindsay} et al.~\cite{Lindsay1997} extend the analysis to handle point-to-point networks; \mbox{Stringel} et al.~\cite{Striegel2000} handle multihop networks; \mbox{Poggi} et al.~\cite{Poggi2003} introduce matrix-DBP scheduling and show the improvement. Related particular cases include the work of Jeffay et al.~\cite{Jeffay2005On-non-preempti} and the work of Quan et al.~\cite{Quan2000}.

In all those papers, the choice/value of the initial $k$-sequences\footnote{Informally, for each task, the $k$-sequence represents the recent past of the task in regard with deadline failures. See Definition~\ref{def:kseq} for a formal definition.} used by DBP to start the scheduling is not really studied. There is no discussion in the literature about the value of the initial $k$-sequence. Very often, the researchers seem to assume that choosing the pattern\footnote{In this document we use the formal language notation, $w^{\alpha}$ means $w$ repeated $\alpha$ times.} $1^k$ for those sequences is optimal regarding the system schedulability without any discussion. Authors also consider very often, e.g., in experimental studies, that examining the first hyper-period of the schedule is significant to conclude. Lastly, the literature, to the best of our knowledge, only proposes sufficient (or necessary) schedulability tests.

\paragraph{This research.} In this research we show that the choice of the initial $k$-sequences is significant and that the period of the schedule can be larger than the hyper-period.
We first show three phenomena: (i) choosing, for each task, the initial $k$-sequence $1^k$ is not optimal, (ii)  we can even start the scheduling from a (fictive) error state (in regard to the initial $k$-sequence) and (iii) the period of feasible DBP-schedules is not necessarily the task hyper-period. We then show that any feasible DBP-schedule is periodic and we upper-bound the length of that period. Lastly, based on our periodicity result we provide an exact schedulability test. From the best of our knowledge, this is the first \emph{exact} (i.e., necessary and sufficient) such test. 

\paragraph{Organization.} Section~\ref{sec:model} presents our model of computation, preliminary definitions and assumptions. Section~\ref{sec:counterexample} presents the three (counter-intuitive) important phenomena, generally ignored by the literature, concerning the initial $k$-sequences. Section~\ref{sec:periodicity} studies the periodicity of DBP schedules. Section~\ref{sec:exacttest} provides a feasibility interval and, based on our periodicity result, our exact schedulability test. Lastly, in Section~\ref{sec:conclusion}, we conclude.

\section{Definitions and assumptions}\label{sec:model}

We consider the scheduling of synchronous periodic task systems. A
system $\tau$ is composed by $n$ periodic tasks $\tau_1,
\tau_2, \ldots, \tau_n$, each task is characterized by a
period $T_i$, a relative deadline $D_i$, an execution requirement 
$C_i$. Such a periodic task generates an
infinite sequence of jobs, with the $k^{\text{th}}$ job
arriving at time-instant $(k - 1)T_i$ ($k = 1, 2,
\ldots$), having an execution requirement of $C_{i}$ units,
and a deadline at time-instant $(k-1)T_{i}+
D_{i}.$ We consider constrained-deadline systems, i.e., $D_{i} \leq T_{i}$. Two additional task characteristics are $m_{i}$ and $k_{i}$ which mandate that at least $m_{i}$ out of any $k_{i}$ consecutive jobs of $\tau_{i}$ must meet its deadline. We consider in this paper a discrete model, i.e., the characteristics of the tasks and the time are integers. 

In order to schedule dynamically such a system, the scheduler will typically based its decision on the $k$-sequence for each task, with the following definition:

\begin{Definition}[$k$-sequence]\label{def:kseq}
The $k$-sequence of task $\tau_{i}$ is a binary string $W = [w_{i,1}, w_{i,2}, \ldots, w_{i,k_{i}}]$ which represents the recent past of the task jobs. By definition a deadline miss corresponds to the value `0' and a deadline met to the value `1', the leftmost bit represents the oldest job.
\end{Definition}

Notice that at time 0, this notion of $k$-sequence is not really defined since the ``past'' is somewhat  undefined, or empty. We assume that initially, at time 0, the scheduler based its decision on an \emph{initial} $k$-sequence for each task; generally the authors consider the string $1^{k_{i}}$ for task $\tau_{i}$, but we will see that is not optimal and causes a loss of generality.

The algorithm DBP assigns a dynamic priority to each active task (say $\tau_{i}$) based on the number of task request[s] which can miss their deadline without violating the $(m_{i},k_{i})$--firm constraint. In other words, DBP bases its decision on the \emph{distance} between the current state ($k_{i}$-sequence) and the closest error state (i.e., where the number of `1' is less than $m_{i}$). For instance, if $\tau_{i}$ is subject to a $(2,3)$--firm constraint and if the current $3$-sequence is $[101]$, the distance is $1$ (the current job of $\tau_{i}$ must be executed) while if the current $3$-sequence is $[011]$ the distance is 2 (the current request of $\tau_{i}$ could be executed, but if not the next one must be). By definition, any $k$-sequence with less than $m_{i}$ `1' violates the $(m_{i}, k_{i})$--firm constraint and corresponds to an error state.

DBP assigns the highest priority to the active task with the smallest distance; variants of DBP are based on the additional rule used to break ties, e.g., EDF-DBP or RM-DBF if we use the Earliest Deadline First or the Rate Monotonic scheduler. We do not consider a specific tie-broker in this study but we assume that the tie-broker is deterministic and memoryless with the following definitions:

\begin{Definition}[Deterministic algorithm]\label{detAlg} 
A scheduling algorithm is said to
be \emph{deterministic} if it generates a unique schedule for any
given set of jobs.
\end{Definition}

\begin{Definition}[Memoryless algorithm]\label{def:memoryless} 
  A non-preemptive scheduling algorithm is said to be \emph{memoryless} if the
  scheduling decision made by it at time $t$ (which corresponds to job arrival or completion) depends only on the static characteristics of active tasks (i.e., $T_{i}, C_{i}, m_{i}, k_{i}$) and on the current state of the system (i.e., the current $k$-sequence and the time elapsed since the last request, of each task).
\end{Definition}

In the previous definition, it may be noticed that since we consider \emph{non-preemptive} systems, the scheduler does not consider the remaining processing time, the latter is always equal to $C_{i}$ for any active task $\tau_{i}$ at scheduling time.

\section{Study of the initial $k$-sequences}\label{sec:counterexample}
\subsection{The non-optimality of the string $1^{k_{i}}$}

In this section we will see that choosing, for each task (say $\tau_{i}$), the initial $k$-sequence to be $1^{k_{i}}$, i.e., consider that the $k_{i}$ previous requests completed by their deadline is not optimal and causes a loss of generality.

\begin{table}
\begin{center}
\begin{tabular}{|lllll|}
\hline
& $T_{i}$ & $C_{i}$ & $m_{i}$ & $k_{i}$\\
\hline
$\tau_{1}$ & 4 & 1 & 2 & 4\\
$\tau_{2}$ & 10 & 8 & 3 & 4\\
\hline
\end{tabular}
\end{center}
\caption{System characteristics.\label{table:example1}}
\end{table}

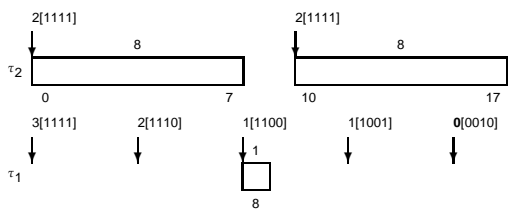
\begin{figure*}
\begin{center}
{\tiny
\setlength{\unitlength}{0.35cm}
\begin{picture}(17,8)
\put(0,1){\makebox(1,1){$\tau_{1}$}}
\put(0,5){\makebox(1,1){$\tau_{2}$}}
\put(1,0){\request{3}{1111}}
\put(1,4){\request{2}{1111}}
\put(1,4){\busy{0}{7}{8}{8}}
\put(5,0){\request{2}{1110}}
\put(9,0){\request{1}{1100}}
\put(9,0){\busy{8}{}{1}{1}}
\put(11,4){\request{2}{1111}}
\put(11,4){\busy{10}{17}{8}{8}}
\put(13,0){\request{1}{1001}}
\put(17,0){\request{\textbf{0}}{0010}}
\end{picture}
}\caption{\label{fig:example1}System not DBP-schedulable choosing the initial $k$-sequence $1^k$: at time 16, $\tau_{1}$ reaches an error state, the $(2,4)$--firm constraint of $\tau_{1}$ being violated.}
\end{center}
\end{figure*}

\begin{Observation}[Non-optimality of $1^{*}$]
Choosing the pattern $1^{k_{i}}$ for the initial $k$-sequence of each periodic task $\tau_{i}$ is not optimal under DBP.
\end{Observation} 

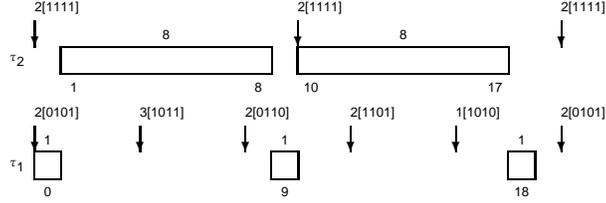
\begin{figure*}
\begin{center}
{\tiny
\setlength{\unitlength}{0.35cm}
\begin{picture}(21,8)
\put(0,1){\makebox(1,1){$\tau_{1}$}}
\put(0,5){\makebox(1,1){$\tau_{2}$}}
\put(1,0){\request{2}{0101}}
\put(1,4){\request{2}{1111}}
\put(1,0){\busy{0}{}{1}{1}}
\put(2,4){\busy{1}{8}{8}{8}}
\put(5,0){\request{3}{1011}}
\put(9,0){\request{2}{0110}}
\put(10,0){\busy{9}{}{1}{1}}
\put(11,4){\request{2}{1111}}
\put(11,4){\busy{10}{17}{8}{8}}
\put(13,0){\request{2}{1101}}
\put(17,0){\request{1}{1010}}
\put(19,0){\busy{18}{}{1}{1}}
\put(21,0){\request{2}{0101}}
\put(21,4){\request{2}{1111}}
\end{picture}
}\caption{\label{fig:example2}The system is DBP-schedulable since all $(m,k)$--firm constraints are met in $[0,20)$ and the system is in the same state at time 20, than at time 10.}
\end{center}
\end{figure*}

\begin{proof}
The proof is based on a (counter-)example which exhibits a system which is not DBP-schedulable choosing the pattern $1^{k_{i}}$ for the initial $k$-sequence of each periodic task $\tau_{i}$ while the system is DBP-schedulable for another initial $k$-sequences.
Our counter-example is described by Table~\ref{table:example1}. Figure~\ref{fig:example1}
shows\footnote{In our figures, picture~{\tiny \setlength{\unitlength}{0.35cm}\requestsymb} represents a task request, we denote by $d[W]$ the corresponding $k$-sequence $W$ and the current distance $d$ from the closest error state.} that the system is not DBP-schedulable if we start to schedule the system with the $k$-sequence $1^{k_{i}}$ for each task $\tau_{i}$. (Notice that Figure~\ref{fig:example1} corresponds also to both RM-DBP and EDF-DBP schedules.) 

On the contrary, with the following $k$-sequences: $[0101], [1111]$ for $\tau_{1}$ and $\tau_{2}$, respectively, the system is DBP-schedulable as illustrated by Figure~\ref{fig:example2}. (Notice that Figure~\ref{fig:example2} corresponds to both RM-DBP and EDF-DBP schedules.)
\end{proof}

\subsection{Starting from an error state}

In the previous section we saw that starting with the $k$-sequence $1^k$ is not always optimal, we also show, with the next example, that we can even start from an error state.

\begin{Example}
We consider this time the very same periodic task system than the one considered in the previous section (i.e., defined by Table~\ref{table:example1}) but we consider the following initial $k$-sequences: $[0010], [1011]$ for $\tau_{1}$ and $\tau_{2}$, respectively. Although we consider an initial fallacious $k$-sequence for $\tau_{1}$ (since $[0010]$ includes a single `1'), Figure~\ref{fig:example3} shows that the system is DBP-schedulable: all $(m,k)$--firm constraints are met (since they are met in $[0,20)$, and since at time 20 we reach the same system state than the schedule illustrated by Figure~\ref{fig:example2} (at time 0), consequently the feasible schedule repeats from time 20).
\end{Example}

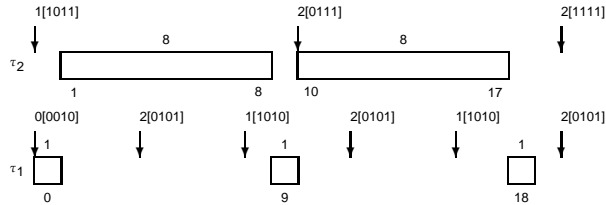
\begin{figure*}
\begin{center}
{\tiny
\setlength{\unitlength}{0.35cm}
\begin{picture}(21,8)
\put(0,1){\makebox(1,1){$\tau_{1}$}}
\put(0,5){\makebox(1,1){$\tau_{2}$}}
\put(1,0){\request{0}{0010}}
\put(1,4){\request{1}{1011}}
\put(1,0){\busy{0}{}{1}{1}}
\put(2,4){\busy{1}{8}{8}{8}}
\put(5,0){\request{2}{0101}}
\put(9,0){\request{1}{1010}}
\put(10,0){\busy{9}{}{1}{1}}
\put(11,4){\request{2}{0111}}
\put(11,4){\busy{10}{17}{8}{8}}
\put(13,0){\request{2}{0101}}
\put(17,0){\request{1}{1010}}
\put(19,0){\busy{18}{}{1}{1}}
\put(21,0){\request{2}{0101}}
\put(21,4){\request{2}{1111}}
\end{picture}
}\caption{\label{fig:example3}The system is DBP-schedulable since all $(m,k)$--firm constraints are met in $[0,20)$ and the fact that the system is in the same state at time 20 than at time 
0 in Figure~\ref{fig:example2}.}
\end{center}
\end{figure*}

We leave open the question of choosing efficiently/optimally the initial $k$-sequences and we assume in the following that those sequences are already chosen but are possibly different from $1^{k_{i}}$.

\section{The periodicity of DBP-schedules}\label{sec:periodicity}

In this section we will study the periodicity of DBP-schedules, an instrumental property to design an exact schedulability test. 

It is quite obvious that any feasible (deterministic and memoryless) schedule is periodic (we will show that property below). Previous examples show that the periodic part of the schedule does not  necessarily start from the origin (e.g., Figure~\ref{fig:example3}). Another important phenomenon is the fact that the period of the schedule is not necessarily equal to the hyper-period $P \equals \lcm\{T_{i} \mid i=1,\ldots,n\}$ as exhibited by the next example.

\begin{Example}
Consider the task characteristics described by Table~\ref{table:example2}. Figure~\ref{fig:example4} shows that the system is DBP-schedulable and that the period of the schedule is \emph{twice} the hyper-period (3). The schedule repeats from time 6 but the period is 6. 
\end{Example}

\begin{table}
\begin{center}
\begin{tabular}{|lllll|}
\hline
& $T_{i}$ & $C_{i}$ & $m_{i}$ & $k_{i}$\\
\hline
$\tau_{1}$ & 3 & 2 & 1 & 3\\
$\tau_{2}$ & 3 & 2 & 1 & 4\\
\hline
\end{tabular}
\end{center}
\caption{System characteristics.\label{table:example2}}
\end{table}

\begin{figure*}
\begin{center}
{\tiny
\setlength{\unitlength}{0.35cm}
\begin{picture}(14,8)
\put(0,1){\makebox(1,1){$\tau_{1}$}}
\put(0,5){\makebox(1,1){$\tau_{2}$}}
\put(1,0){\request{3}{111}}
\put(1,4){\request{3}{111}}
\put(1,0){\busy{0}{1}{2}{2}}

\put(4,0){\request{3}{111}}
\put(4,4){\request{2}{110}}

\put(4,4){\busy{4}{5}{2}{2}}

\put(7,0){\request{2}{110}}
\put(7,4){\request{3}{101}}

\put(7,0){\busy{6}{7}{2}{2}}

\put(10,0){\request{3}{101}}
\put(10,4){\request{2}{010}}

\put(10,4){\busy{9}{10}{2}{2}}

\put(13,0){\request{2}{010}}
\put(13,4){\request{3}{101}}

\put(13,0){\busy{12}{13}{2}{2}}

\put(16,0){\request{3}{101}}
\put(16,4){\request{2}{010}}
\end{picture}
}\caption{\label{fig:example4}The system is DBP-schedulable since all $(m,k)$--firm constraints are met in $[0,15)$ and the system is in the same state at time 15 than at time 9.}
\end{center}
\end{figure*}
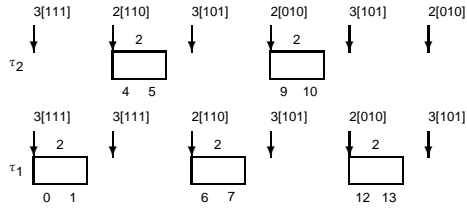

Now we will prove the schedule periodicity and (upper-)bound the length of that period. That result (and its proof) is instrumental to design our \emph{exact} schedulability test for the DBP algorithm (see Algorithm~\ref{algo:exact}).

\begin{Theorem}[Period upper-bound]\label{theorem:periodicity}
Any feasible DBP-schedule of a synchronous $(m,k)$--firm periodic task set is finally periodic. Moreover, the period of the schedule is a multiple of $P$, upper-bounded by $\prod_{i=1}^{n} \sum_{j=m_{i}}^{k_{i}} {k_{i}\choose j} \times P$.
\end{Theorem}

\begin{proof}
DBP is an on-line scheduling algorithm, i.e., it takes its decisions based on the (static) task characteristics and the (dynamic) current state of the system: the current $k$-sequence and the time elapsed since the last request, of each task. Since we consider \emph{synchronous} periodic, \emph{constrained-deadline}, and  DBP-feasible task systems, we know that at each instant $t = k \cdot P$ ($k \in \mathbb{N}$) all requests which occur strictly before $t$ have completed their execution. Moreover, a new request of each task occurs at time $t$. Consequently, regarding the time elapsed since the last request, the system is in the same state at time $0, P, 2P\ldots$. The only difference (if any) concerns the $k$-sequences. For a given task (say $\tau_{i}$) we can distinguish between $\sum_{j=m_{i}}^{k_{i}} {k_{i}\choose j}$ distinct $k$-sequences (i.e., $k$-sequences with at least $m_{i}$ `1' of length $k_{i}$ ). Consequently, in the worst case we have to consider $\prod_{i=1}^{n} \sum_{j=m_{i}}^{k_{i}} {k_{i}\choose j}$ hyper-periods before reaching a state already considered in the past. The property follows from the fact we consider a deterministic and memoryless scheduler.
\end{proof}

\section{Exact schedulability test}\label{sec:exacttest}

Now we have the material to provide an exact schedulability test (and an algorithm) for any deterministic and memoryless DBP scheduler. Notice that we assume that the initial $k$-sequences are fixed.

\subsection{Feasibility interval}

\begin{Definition}[Feasibility interval]
  A \emph{feasibility interval} is a finite interval such that if no $(m,k)$--firm constraint is missed  while considering only requests within this interval then no
  $(m,k)$--firm constraint will ever be missed.
\end{Definition}

\begin{Corollary}\label{coro:interval}
The interval $\left[0, \prod_{i=1}^{n} \sum_{j=m_{i}}^{k_{i}}{k_{i}\choose j}\times P\right)$ is a feasibility interval for the scheduling of synchronous constrained-deadline $(m,k)$--firm periodic systems using DBP.
\end{Corollary}

\begin{proof}
It is a direct consequence of the proof of Theorem~\ref{theorem:periodicity}: in the worst case the system state at time $0, P, 2P, \left[(\prod_{i=1}^{n} \sum_{j=m_{i}}^{k_{i}}{k_{i}\choose j}) - 1\right] P$ are different but the system state at time $\prod_{i=1}^{n} \sum_{j=m_{i}}^{k_{i}}{k_{i}\choose j} \times P$ by Theorem~\ref{theorem:periodicity} occurred already in the past.
\end{proof}

\subsection{Exact Algorithm}

Based on Corollary~\ref{coro:interval} a straightforward exact schedulability test consists in building the schedule (by means of simulation) in the time interval $[0, \prod_{i=1}^{n} \sum_{j=m_{i}}^{k_{i}}{k_{i}\choose j} \times P)$. We designed a \emph{faster} (in average) algorithm (Algorithm~\ref{algo:exact}) which stops once the schedule repeats (or a failure occurs). The idea is to compare, at each hyper-period, the system state with the previous hyper-period states. Algorithm~\ref{algo:exact} assumes that the function \textbf{Schedule} stops the simulation and returns false once a $(m,k)$--firm constraint is violated. The \textbf{system-state} considered in Algorithm~\ref{algo:exact} is the current $k$-sequence for each task.

\begin{algorithm}
\SetKw{kwschedule}{Schedule}
\SetKw{kwstate}{system-state}
\KwIn{task set $\tau$}
\KwOut{feasible}
\KwData{State St[max]} 
\tcc{$\max = \prod_{i=1}^{n} \sum_{j=m_{i}}^{k_{i}}{k_{i}\choose j}$}
\tcc{array St is indexed from 0}
\Begin{
\tcc{we save the current system state}
St[0] := \kwstate \;
$\ell$ :=  0\;
current-time := 0 \;
periodic := false \;
feasible := true \;

\Repeat{(periodic $\vee$ $\neg$feasible)}{
feasible := \kwschedule\ from current-time to current-time + $P$ \;
current-time := current-time + $P$ \;
\eIf{$\ell = \max - 1$}{
  \For{$j$ := 0 \KwTo $max - 2$}{
    St[$j$] := St[$j+1$]\;
  }
}{$\ell$ :=  $\ell + 1$\;}

St[$\ell$] := \kwstate \;
$j$ := $\ell - 1$ \;
\While{($j>0$ $\wedge$ $\neg$periodic)}{
  periodic := (St[$j$] = St[$\ell$]) \;
  $j$ := $j - 1$\;
}
}

\Return{feasible}\;
}
\caption{Exact DBP-schedulability test.\label{algo:exact}}
\end{algorithm}

The worst-case time complexity of the exact schedulability test is $\bigOh(\prod_{i=1}^{n} \sum_{j=m_{i}}^{k_{i}}{k_{i}\choose j}\times P)$. We believe that for many real real-time applications the time complexity of our \emph{off-line} test is reasonnable, e.g., if we consider harmonic task periods and limited $(m,k)$ parameters.

\section{Conclusion}\label{sec:conclusion}
In this study we showed that the choice of the initial $k$-sequences \emph{is} significant regarding the system schedulability. We exhibited three phenomena: (i) choosing, for each task, the initial $k$-sequence $1^k$ is not optimal, (ii) we can even sometimes start the scheduling from a (fictive) error state and (iii) the period of feasible DBP-schedules is not necessarily the task hyper-period. We then showed that any feasible DBP-schedule is periodic and we upper-bounded the length of that period. Lastly, based on our periodicity result we provided an exact schedulability test. 

\paragraph{Future work.} We left open the question of choosing efficiently/optimally the initial $k$-sequences. We left also open the following question: does the initial $k$-sequence impact the period of the (feasible) schedule? This research could be also extended to consider sporadic or asynchronous periodic task sets.

\paragraph{Acknowledgments.} The author would like to thank \mbox{Qiong} Ye \mbox{Song} and \mbox{Raymond} \mbox{Devillers} for taking part in interesting discussions. Finally, comments of anonymous reviewers helped improving the presentation of the paper.

\bibliographystyle{acm}
\bibliography{../doc.bib}

\begin{thebibliography}{1}

\bibitem{Hamdaoui1995}
{\sc Hamdaoui, M., and Ramanathan, P.}
\newblock A dynamic priority assignment technique for streams with
  $(m,k)$--firm deadlines.
\newblock {\em IEEE Transactions on Computers 44}, 12 (December 1995),
  1443--1451.

\bibitem{Hamdaoui1997}
{\sc Hamdaoui, M., and Ramanathan, P.}
\newblock Evaluating dynamic failure probability for streams with $(m,k)$--firm
  deadlines.
\newblock {\em IEEE Transactions on Computers 46}, 12 (December 1997),
  1325--1337.

\bibitem{Jeffay2005On-non-preempti}
{\sc Jeffay, K., Stanat, D., and Martel, C.}
\newblock On non-preemptive scheduling of periodic and sporadic tasks.
\newblock In {\em Real-Time Systems Symposium\/} (1991).

\bibitem{Li2004}
{\sc Li, J., Song, Y.~Q., and Simonot-Lion, F.}
\newblock Schedulability analysis for systems under $(m,k)$--firm constraints.
\newblock In {\em Factory Communication Systems\/} (September 2004), IEEE
  Computer Society, pp.~23--30.

\bibitem{Lindsay1997}
{\sc Lindsay, W., and Ramanathan, P.}
\newblock Dbp-m: a technique for meeting end-to-end $(m,k)$--firm guarantee
  requirements in point-to-point networks.
\newblock In {\em Annual conference on local computer networks\/} (1997), I.~C.
  Society, Ed., pp.~294--303.

\bibitem{Poggi2003}
{\sc Poggi, E., Song, Y.~Q., Koubaa, A., and Wang, Z.}
\newblock Matrix-dbp for $(m, k)$--firm real-time guarantee.
\newblock In {\em Real-Time and Embedded Systems\/} (Paris, France, 2003),
  Y.~Trinquet, Ed., pp.~457--480.

\bibitem{Quan2000}
{\sc Quan, G., and Hu, X.}
\newblock Enhanced fixed-priority scheduling with $(m,k)$--firm guarantee.
\newblock In {\em Real-time systems symposium\/} (2000), I.~C. Society, Ed.,
  pp.~79--88.

\bibitem{Striegel2000}
{\sc Striegel, A., and Manimaram, G.}
\newblock Best-effort scheduling of $(m,k)$--firm real-time streams in multihop
  networks, 2000.

\end{thebibliography}
\end{document}